\documentclass[letterpaper, 10 pt, journal, twoside]{IEEEtran}  %

 \newcommand{\IEEEwarning}{
  \onecolumn
  \noindent \textcopyright{} 2018 IEEE. Personal use of this material is permitted. Permission from IEEE must be obtained for all other uses, in any current or future media, including reprinting / republishing this material for advertising or promotional purposes, creating new collective works, for resale or redistribution to servers or lists, or reuse of any copyrighted component of this work in other works

  \bigskip \noindent Article submitted to \textit{IEEE Transactions on Control Systems Technology}.
  \twocolumn
  \newpage
}

\usepackage{graphics} 
\usepackage{epsfig} 
\graphicspath{ {images/} }
\usepackage{color,soul} 

\usepackage{times} 
\usepackage{cite}

\usepackage{amsmath} 
\usepackage{amssymb}

\usepackage{bbm}
\usepackage{amsthm}                               
\theoremstyle{remark}
\newtheorem{remark}{Remark}
\newtheorem{theorem}{Theorem}
\newtheorem{lemma}{Lemma}

\newtheorem{definition}{Definition}
\usepackage{enumerate}

\title{Quickest Detection of Intermittent Signals With Application to Vision Based Aircraft Detection
}

\author{Jasmin James,  Jason J. Ford and Timothy L. Molloy
\thanks{Parts of this paper have been presented at the 2017 Asian Control Conference \cite{James2017b}. The authors are with the School of Electrical Engineering and Computer Science, Queensland University of Technology, 2 George St, Brisbane QLD, 4000 Australia. {\tt\small jasmin.james@connect.qut.edu.au, j2.ford@qut.edu.au, t.molloy@qut.edu.au}. This work was supported by funding from the Australian Research Council Centre of Excellence CE140100016 in Robotic Vision.}%
}

\begin{document}
\IEEEwarning
\maketitle
\thispagestyle{empty}
\pagestyle{empty}
\begin{abstract}
In this paper we consider the problem of quickly detecting changes in an intermittent signal that can (repeatedly) switch between a normal and an anomalous state. We pose this intermittent signal detection problem as an optimal stopping problem and establish a quickest intermittent signal detection (ISD) rule with a threshold structure. We develop  bounds to characterise the performance of our ISD rule and establish a new filter for estimating its detection delays. Finally, we examine the performance of our ISD rule in both a simulation study and an important vision based aircraft detection application where the ISD rule demonstrates improvements in detection range and false alarm rates relative to the current state of the art aircraft detection techniques. 

\end{abstract}
\begin{IEEEkeywords}
Change Detection, Bayesian Quickest Change Detection, Sense and Avoid, Filtering
\end{IEEEkeywords}

\section{Introduction}
\IEEEPARstart{Q}{uickly} detecting the presence of an anomaly condition that can repeatedly appear and disappear is important in many applications such as fault detection \cite{Hwang}, cyber-security \cite{Tartakovsky2012}, intrusion or anomaly detection \cite{Tartakovsky2012}, and vision based aircraft detection \cite{Lai2013}. In vision based aircraft detection, this anomaly condition represents the potential emergence of an aircraft anywhere in an image which needs to be quickly detected for collision avoidance purposes.  We describe a signal containing this repeating anomaly condition as an intermittent signal. In this paper we aim to pose and solve this quickest intermittent signal detection (ISD) problem in a Bayesian setting that allows us to trade off average detection delay and false alarm probability.

In classic Bayesian quickest change detection, it is assumed that a permanent change in the statistics of a sequence of random variables occurs at some random unknown change time \cite{Tartakovsky2014}. The classic Bayesian  criterion seeks to minimise the average detection delay subject to a constraint placed on the probability of a false alarm. For this Bayesian formulation, Shiryaev established an optimal stopping rule which compares the posterior probability of a change to a threshold \cite{Tartakovsky2014}.   

Inspired by  classic  Bayesian quickest change detection, several alternative quickest detection problems have been posed in the last decade. Incipient fault detection seeks to identify slow drifts in system parameters \cite{Roychoudhury}; multi-cyclic detection seeks to identify a distant change in a stationary regime where detection procedures are reset after each false alarm \cite{Polunchenko2011}; quickest transient detection seeks to identify a change that occurs once for a period of time and then disappears\cite{Broder,Guepie2017}; and quickest detection under transient dynamics that seeks to identify a persistent change which does not happen instantaneously, but after a series of transient phases \cite{2017Zou}. In this paper, we consider a new quickest ISD problem where a change can repeatedly appear and disappear over time.     

Our quickest ISD problem is inspired by the important vision based aircraft detection application in which a small pixel-sized aircraft can visually emerge  anywhere in an image and can potentially transition in and out of view. Previous detection solutions have utilised \textit{ad hoc} maximum likelihood approaches \cite{Lai2011,Lai2013,molloy2017below}, and methods of non-Bayesian quickest change detection \cite{James2017a}. Here, we instead pose a quickest ISD problem and seek an optimal detection rule with the goal of quickly detecting when an aircraft emerges in an image sequence.  

The key contributions of this paper are:
\begin{enumerate}[i)]
\item Posing the quickest ISD problem and utilising an optimal stopping framework to establish an ISD rule with a threshold structure;
\item Introducing a new occupation time filter to estimate the detection delay of our ISD rule;
\item Experimentally demonstrating the improvements offered by our ISD rule in the vision based aircraft detection application.
\end{enumerate}

The rest of this paper is structured as follows. In Section \ref{sec:probForm} we pose our quickest ISD problem and associated cost criterion. In Section \ref{sec:ost} we establish an optimal ISD rule. In Section \ref{sec:perform} we provide performance characteristics for our ISD rule.
In Section \ref{sec:sim} we examine the performance of our ISD rule in a simulation study. In Section \ref{sec:app} we apply our ISD rule to vision based aircraft detection and examine its performance on an experimentally captured flight dataset. 
\section{Problem Formulation}\label{sec:probForm}
Let $X_k \in \{e_1, e_2\}$, for $k \geq 0$, be a sequence of random variables representing an intermittent signal that switches between a normal state $e_1$ and an anomalous state $e_2$ at (unknown) random time instances. Here $e_i \in \mathbb{R}^{2}$ are indicator vectors with $1$ as the $i$th element and zeros elsewhere.
For $k >0$, the intermittent signal $X_k$ is hidden within measurements $y_k \in \mathbb{R}^M$, that are an independent and identically distributed (i.i.d.) sequence of random variables with (marginal) probability density functions $f^1(y_k)$ when $X_k = e_1$ and $f^2(y_k)$ when $X_k = e_2$.

In this paper, we shall assume that the intermittent signal $X_k$ is a first-order time-homogeneous Markov chain.
Let $\rho$ be the probability of transitioning from the normal state behavior $e_1$ to the anomalous state behavior $e_2$, and let $a$ be the probability of self transition for $e_2$. 
Let us denote the transition probabilities at each time instant by $A^{i,j} \triangleq P(X_{k+1} = e_i | X_k = e_j)$ for $i,j \in \{1,2\}$ as
\begin{equation}
A = \begin{bmatrix}
1-\rho & 1- a \\
\rho & a
\end{bmatrix}.
\end{equation}

For $k\geq 0$, we can describe the intermittent signal (state process) $X_{k}$, as follows
\begin{eqnarray} 
X_{k+1} = AX_{k} +V_{k+1}\label{eqn:xprocess} 
\end{eqnarray}
where $V_{k+1} \in \mathbb{R}^{2\time1}$ is a martingale increment and the initial state $X_0$ has distribution $\hat{X}_0$. 
For the remainder of this paper we will define  $X_{[0,k]} \triangleq \{X_0, \dots ,X_k \}$ and $y_{[1,k]} \triangleq \{y_1, \dots ,y_k \}$ as shorthand for sequences of these random variables.

We now introduce a probability measure space used to pose our quickest ISD problem. Similar to \cite{ford2011} we consider the set $\Omega$ consisting of all infinite sequences $\omega \triangleq (X_{[0,k]}, \dots; y_{[1,k]}, \dots)$. Since $\Omega$ is separable and a complete metric space it can be endowed with a Borel $\sigma-$algebra $\mathcal{F} \triangleq \mathcal{B}_\Omega$. Using Kolmogorov's extension theorem we can now define a probability measure $\mathbb{P}$ on $(\Omega,\mathcal{F})$. We let $\mathbb{E}$ denote the expectation operation under the probability measure $\mathbb{P}$.

In this paper, our goal is to quickly detect when  $X_k$ is in $e_2$ by seeking to design a stopping rule $\tau \geq 0$ that minimises the following ISD cost criterion
\begin{equation}\label{eqn:cost}
J(\tau) = \mathbb{E} \Bigg[c \sum_{\ell=0}^{\tau -1} \langle X_\ell,e_2 \rangle + \langle X_\tau,e_1 \rangle  \Bigg],
\end{equation}
where $\langle . ,. \rangle$ denotes the inner product and $c$ is the penalty for the total amount of time spent in state $e_2$ before declaring an alert at $\tau$. This ISD cost criterion represents our desire to detect being in state $e_2$ as quickly as possible whilst avoiding false alarms (that is, avoid incorrectly declaring a stopping alert when the state is  $e_1$).

\section{Intermittent Signal Detection: Optimal Stopping time}\label{sec:ost}
In this section we first establish an equivalent representation of our ISD cost criterion in terms of the conditional mean estimates (CMEs) of the intermittent signal $X_k$. We then pose our quickest ISD problem as an optimal stopping problem and establish an optimal ISD rule that has a test statistic with a threshold structure. Finally we present the hidden Markov model (HMM) filter that can be used to efficiently calculate this test statistic.

\subsection{Equivalent Representation of the ISD Cost Criterion}
\begin{lemma}\label{lem:isdCost}
The ISD cost criterion
 \eqref{eqn:cost} can be expressed in terms of the CMEs of the intermittent signal $X_k$ as
\begin{equation} \label{eqn:Xhatcost}
J(\tau) = \mathbb{E} \Bigg[ c \sum_{\ell=0}^{\tau -1}    \hat{X}_\ell^2 + \hat{X}_\tau^1  \Bigg]
\end{equation}
where $\hat{X}_k^i \triangleq \mathbb{E}\big[X_k^i \big|y_{[1,k]} \big]$ denotes the CME of being in state $e_i$ given the measurements $y_{[1,k]}$. 

Moreover for the case where $e_2$ is an absorbing state then the transition probability $a=1$ and the  cost criterion \eqref{eqn:cost} reduces to the classic Bayesian quickest change detection criterion \cite{Shiryaev}
\begin{equation}
J^0(\tau) = \ c  \mathbb{E} \bigg[(\tau -\nu)^+\bigg] + \mathbb{P} (\tau < \nu), 
\end{equation}
where $\nu$ is the time of transition into the absorbing state $e_2$  and $(\tau -\nu)^+ \triangleq \max(0,\tau -\nu)$.
\end{lemma}

\begin{proof}
The ISD cost criterion \eqref{eqn:cost} can be expressed as
\[
\mathbb{E} \Bigg[c \sum_{\ell=0}^{\tau -1} \langle X_\ell,e_2 \rangle + \langle X_\tau,e_1 \rangle  \Bigg] = \mathbb{E} \Bigg[   c \sum_{\ell=0}^{\tau -1}   X_\ell^2 + X_\tau^1 \Bigg].
\]

Following \cite{krishnamurthy2016} and using the tower rule for conditional expectations \cite[pg. 331]{elliott1995} we obtain
\[
\begin{split}
\mathbb{E} \Bigg[ c \sum_{\ell=0}^{\tau -1}   X_\ell^2 + X_\tau^1 \Bigg] 
 =& \mathbb{E} \left[ \sum_{\ell=0}^{\tau -1} \mathbb{E} \left[ \left. c X_\ell^2
 + 
 X_\tau^1  \right| y_{[1,k]} \right] \right] \\
=& \mathbb{E} \Bigg[c \sum_{\ell=0}^{\tau -1}    \hat{X}_\ell^2 + \hat{X}_\tau^1  \Bigg].
\end{split}
\]
This proves the first lemma result. 

For the second result, if $e_2$ is an absorbing state, then $a=1$ and once $X_k=e_2$ it remains in $e_2$ and the following holds,
\begin{equation*}
\mathbb{E} \Bigg[c \sum_{\ell=0}^{\tau -1} \langle X_\ell,e_2 \rangle \Bigg] = c  \mathbb{E} \Big[(\tau -\nu)^+\Big] 
\end{equation*}
and
\begin{equation*}
 \mathbb{E} \Bigg[ \langle X_\tau,e_1 \rangle \Bigg] = \mathbb{P}(\tau < \nu).
\end{equation*}
This completes the proof.
\end{proof}

Lemma \ref{lem:isdCost} shows that our quickest ISD problem is a generalisation or relaxation of the Bayesian quickest detection problem, in the sense that, when $e_2$ becomes an absorbing state the ISD cost criterion \eqref{eqn:cost} reduces to the classic Bayesian quickest detection problem \cite{krishnamurthy2016}. 
Further, this  lemma  allows our proposed ISD cost criterion \eqref{eqn:cost} to be expressed in terms of the CMEs which will now be  used to establish an optimal ISD rule.




\subsection{Optimal ISD Rule}
In the following  theorem we show that an optimal solution for the ISD cost criterion is a stopping rule with a threshold structure. 
\begin{theorem}
For the ISD cost criterion \eqref{eqn:Xhatcost}, there is 
an optimal ISD rule with stopping time $\tau^*$, and threshold point $h_s \geq 0$ given by
\begin{equation}\label{eqn:stop} 
\tau^* = \inf\{k \geq 0 : \hat{X}^2_k \geq h_s\}. 
\end{equation}
\end{theorem}
\begin{proof}
In a slight abuse notation, we let $\mathbb{E} \left[  \cdot \left| \hat{X} \right] \right. $ denote the expectation operation corresponding to the probability measure where the initial state $X_0$ has distribution $\hat{X}$. We then define a cost criterion for different initial distributions as
\begin{equation}
\bar{J}(\tau,\hat{X}) \triangleq \mathbb{E} \left.  \left[ c \sum_{\ell=0}^{\tau -1}    \hat{X}_\ell^2 + \hat{X}_\tau^1 \right|   \hat{X}  \right].
\end{equation}
Noting that $J(\tau)=\bar{J}(\tau,\hat{X}_0)$, we can define a value function $V(\hat{X}_k) \triangleq \min_\tau \{ \bar{J}(\tau,\hat{X}_k)\}$ for our ISD cost criterion \eqref{eqn:Xhatcost} described by the recursion  \cite[pg. 156 and 258]{krishnamurthy2016}
\begin{equation}\label{eqn:val}
\begin{split}
V(\hat{X}_k) = \min \bigg\{c \hat{X}_k^2  
+ \mathbb{E} \left[ V\left(\hat{X}^+(\hat{X}_k,y)\right) \bigg| \hat{X}_k  \right] ,  \hat{X}_k^1  \bigg\},
\end{split}
\end{equation}
where $\hat{X}^+(\hat{X},y) = \langle \underline{1},B(y) A \hat{X} \rangle^{-1} B(y)A \hat{X} $, and $B(y) = \text{diag}(f^1(y), f^2(y))$.

Let $\mathcal{S} \triangleq \{ \hat{X}^2_k : V(\hat{X}_k) = \hat{X}_k^1\}$ denote the optimal stopping set that we are seeking. 	
Similar to the approach in \cite[sec. 12.2.2]{krishnamurthy2016}, 
noting that the value function $V(\hat{X}_k)$ is concave,  Theorem 12.2.1 of \cite{krishnamurthy2016} gives that the stopping set $\mathcal{S} \subset [0,1]$  is convex. This implies that  $\mathcal{S}$  is an interval of the form $[h_s,d]$ for some $0 \leq  h_s \leq d \leq 1$.

We now write our value function \eqref{eqn:val}   when $\hat{X}_k=e_2$  as 
\begin{equation*}
V(e_2)  = \min \left\{ c + \mathbb{E} \left[ V\left(\hat{X}^+(e_2,y)\right)\bigg|  e_2 \right],   0  \right\}.
\end{equation*}
Since $V\left(\hat{X}^+(e_2,y)\right)$ is positive  then
$V(e_2)  = 0$, which shows $\hat{X}_k^2=1$ belongs to the stopping set, thus $d=1$ and $\mathcal{S}$ is an interval of the form  $[h_s,1]$. We can express the optimal  stopping time as the first time that the stopping set $\mathcal{S}$ is reached, in the sense that
\begin{equation*}
\begin{split}
\tau^*  &= \inf\{k \geq 0  : \hat{X}^2_k \geq h_s\}.
\end{split}
\end{equation*}
This completes the proof. 
\end{proof} 
We note that when $a=1$, perhaps unsurprisingly, this ISD rule reduces to Shiryaev's Bayesian detection (SBD) rule \cite{Tartakovsky2012}.

While it is possible to write down a dynamic programming equation for the optimal threshold for the stopping time \eqref{eqn:stop}, in practice the threshold $h_s$ is selected to trade off the alert delay and the probability of a false alarm. 

\subsection{HMM Filter}
We now present the HMM filter for \eqref{eqn:xprocess} which can be used to efficiently calculate the CME $\hat{X}_k = \mathbb{E} \big[X_k \big|y_{[1,k]} \big]$, where the $2$nd element $\hat{X}_k^2$ can be used to implement  our ISD  rule \eqref{eqn:stop}.

At time $k>0$, we let $B(y_k) = \text{diag}(f^1(y_k), f^2(y_k))$  denote the diagonal matrix of output probability densities.  We can now calculate the CME  $\hat{X}_k $ at time $k$,  via the HMM filter   \cite{elliott1995}
\begin{equation}\label{eqn:CME}
\hat{X}_k= N_{k} B(y_k) A \hat{X}_{k-1},
\end{equation}
with initial condition $\hat{X}_{0}$ and where $N_{k}$ are scalar normalisation factors defined by 
\begin{equation}\label{eqn:Nk}
N_{k}^{-1} \triangleq \langle \underline{1},B(y_k) A \hat{X}_{k-1} \rangle.
\end{equation}

\section{Performance Bound and Delay Estimation}\label{sec:perform}
In this section we provide a bound for the probability of a false alarm (PFA) for our ISD rule. We then propose a new occupation time filter that can be used to estimate how long has been spent in the anomalous state $e_2$ before an alert is declared. Finally we establish some stability results for our proposed occupation time filter. 

\subsection{Bound on Probability of False Alarm}
For a given threshold $h_s$ used in our ISD rule \eqref{eqn:stop}, we define the PFA as the probability that the system is in the normal state $e_1$ when an alert is declared, that is $ \text{PFA}(\tau) \triangleq \mathbb{P} \left( X_\tau = e_1 \right)$. We can then bound the PFA as follows
\begin{equation*}\label{eqn:pfa}
\begin{split}
\text{PFA}(\tau)
&= \mathbb{P} \left( X_\tau = e_1 \right)\\
&= 1- \mathbb{P} \left( X_\tau = e_2 \right)\\
&= 1 - \mathbb{E} \left[\hat{X}^2_\tau \right] \\
&\leq 1-h_s.
\end{split}
\end{equation*}
In the second line we have followed \cite{molloy2016} and used the tower rule for conditional expectations \cite[pg. 331]{elliott1995}. In the third line we have used the fact that the CME $\hat{X}^2_\tau =  \mathbb{P} \big( X_\tau = e_2 \big| y_{[0,\tau]} \big)$. Finally we use the definition of the stopping time \eqref{eqn:stop}.

\subsection{Occupation Time Filter}
At time $k>0$, for $i=\{1,2\}$ we define the state occupation time ${\mathcal{O}}^i_{k} \in \mathbb{R}$  as
\begin{equation}
\begin{split}
{\mathcal{O}}^i_{k} & \triangleq \sum_{n=0}^{k-1} \langle X_n, e_i \rangle.
\end{split}
\end{equation}
We also define  the CME of the occupation time $\hat{\mathcal{O}}^{i}_{k} \triangleq \mathbb{E} \big[ \mathcal{O}_k^i \big| y_{[1,k]}\big]$, and  the CME of the occupation time ending in state $X_k$ as $\hat{\mathcal{O}}^{i,X}_{k} \triangleq \mathbb{E} \big[ \mathcal{O}_k^i X_k \big| y_{[1,k]}\big] \in \mathbb{R}^2$. 
Filters for these CMEs are established in the next lemma.
\begin{lemma}\label{lem:occ}
For $k> 0$, an occupation time filter $\hat{\mathcal{O}}^{i}_{k}$  for all $i \in \{ 1,2\}$ is given by
\begin{equation}\label{eqn:stateocc}
\begin{split}
\hat{\mathcal{O}}^{i,X}_k&=N_k B(y_k) A (\hat{\mathcal{O}}^{i,X}_{k-1} + \hat{X}^i_{k-1} e_i),\\
\hat{\mathcal{O}}^{i}_k &= \langle \underline{1}, \hat{\mathcal{O}}^{i,X}_k \rangle,
\end{split}
\end{equation}
with initial conditions  $\hat{\mathcal{O}}^{i,X}_0 = \underline{0}$ and $\hat{X}_k$ is given by \eqref{eqn:CME}.
\end{lemma}
\begin{proof}
See appendix for proof.
\end{proof}

By setting $i=2$, this lemma lets us estimate how long has been spent in the anomalous state $e_2$ when our ISD rule declares an alert. We highlight that similar  occupation time filters are presented in \cite{elliott1995} for a delayed measurement model.

\subsection{Occupation Time Filter Stability}
We now present results characterising the stability of our proposed occupation time filter with respect to initial conditions. 

We first introduce some required concepts before we present our proof. Let $\hat{\mathcal{O}}^{i}_k(\hat{X}_0)$ and $\hat{X}_k(\hat{X}_0)$ denote the occupation time CME filter and the HMM filter respectively with initial condition $\hat{X}_0$. 
We now define an average error rate $\mathcal{R}^E_k(\hat{X}_0,\check{X}_0)$
between correct and misspecified initial conditions $\hat{X}_0$ and $\check{X}_0$, respectively, as
\begin{equation}
\mathcal{R}^E_k(\hat{X}_0,\check{X}_0) \triangleq \frac{ \left| \hat{\mathcal{O}}^{i,X}_k(\hat{X}_0) -\hat{\mathcal{O}}^{i,X}_k(\check{X}_0) \right|}{k}.
\end{equation}
Finally, a function $\phi$ is said to be of class $\mathcal{K}$ if it is strictly increasing, continuous and $\phi(0)=0$. A function $\beta$ is said to be of class $\mathcal{K}\mathcal{L}$ if for each $t \geq 0$,  $\beta(\cdot,t)$ is of class $\mathcal{K}$, and for each $s>0$, $\beta(s,\cdot)$ is decreasing to zero. 

\begin{definition}\label{def:beta}
(\textit{Asymptotic stability with respect to initial conditions}) The HMM filter $\hat{X}_k(\cdot)$  is said to be asymptotically stable with respect to initial conditions if there exists a $\beta(\cdot,\cdot) \in \mathcal{KL}$  such that for any $\hat{X}_0$ and $\check{X}_0$, 
\begin{equation}\label{eqn:err}
| \hat{X}_k(\hat{X}_0)-\hat{X}_k(\check{X}_0) | \leq \beta(| \hat{X}_k(\hat{X}_0)-\hat{X}_k(\check{X}_0) |,k). 
\end{equation}

\begin{lemma}
Assume that the HMM filter $\hat{X}_k(\cdot)$ is asymptotically stable with respect to initial conditions. Then,  the occupation time filter \eqref{eqn:err}  is  (average error rate) practically stable in the sense that for any given 
 $0<\delta \leq 1$, there is a $H$ such that for all $k>H$, and for any $\hat{X}_0$ and $\check{X}_0$, we have
\begin{equation}
\mathcal{R}^E_k(\hat{X}_0,\check{X}_0) \leq \delta + \frac{H}{k}.
\end{equation}
\end{lemma}
\begin{proof}
See appendix for proof.
\end{proof}

\end{definition}
\begin{remark} \label{rem:shue}
 There are standard mild conditions under which the HMM filter $\hat{X}_k(\cdot)$ is asymptotically stable in the sense of Definition \ref{def:beta}, see \cite{Shue1998} for more information. 
\end{remark}
\section{Simulation Results}\label{sec:sim}
In this section we examine the performance of our ISD rule \eqref{eqn:stop} and occupation time filter \eqref{eqn:stateocc} in simulation. 

\subsection{Illustrative Example of the ISD and SBD Optimal Stopping Rules} 
We simulated a hand-crafted intermittent signal $X_k$ which switched between normal $e_1$ and anomalous $e_2$ states. The measurements $y_k$ are  i.i.d. with marginal probability  densities $f^1(y) = \psi(y-1)$ when $X_k =e_1$ and $f^2(y) = \psi(y-2)$ when $X_k = e_2$, where $\psi (\cdot)$ is zero-mean Gaussian probability density function with variance $\sigma^2$. 
The ISD rule \eqref{eqn:stop} with $\rho=0.01$ and $a=0.99$, and the SBD rule \eqref{eqn:stop} with $\rho=0.01$ and $a=1$ were both applied to the simulated observation data with $\sigma^2 =5$.

From top to bottom, Figure \ref{fig:simcomp} gives an illustrative example of the intermittent signal $X_k$, the measurements $y_k$, and a comparison of the ISD and the SBD test statistics against a  threshold of $h_s = 0.7$. In this example the underlying intermittent signal switches into the anomalous state at $k=600$. Our ISD rule (correctly) declares an alert at $k=617$ with no false alarms.  The SBD test statistic also exceeds the threshold at $k=617$, however the SBD rule declares an alert at $k=223$ corresponding to a false alarm. 

\begin{figure}
\begin{center}
\includegraphics[scale=0.63]{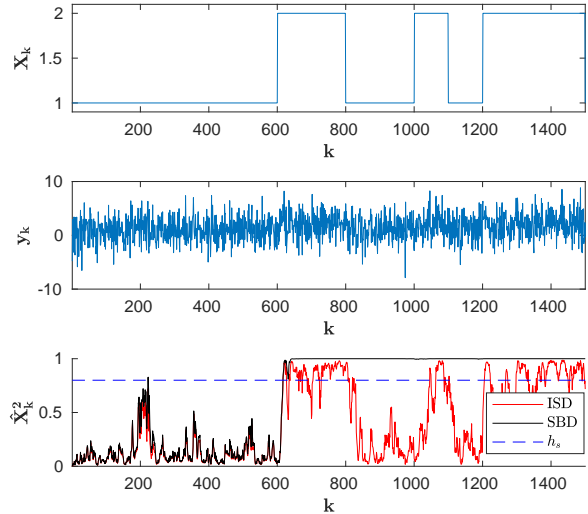}
  \caption{ From top to bottom: the intermittent signal $X_k$, the measurements $y_k$ and a comparison of the ISD and the SBD optimal stopping rules for an arbitrarily  selected threshold of $h_s = 0.7$. Our ISD rule alerts at $k=617$, while the SBD rule  alerts at $k=223$ (corresponding to a false alarm).  }
\label{fig:simcomp}
\end{center}
\end{figure}

\subsection{Performance of Stopping rules in Monte Carlo Study}

\begin{figure}
\begin{center}
\includegraphics[scale=0.6]{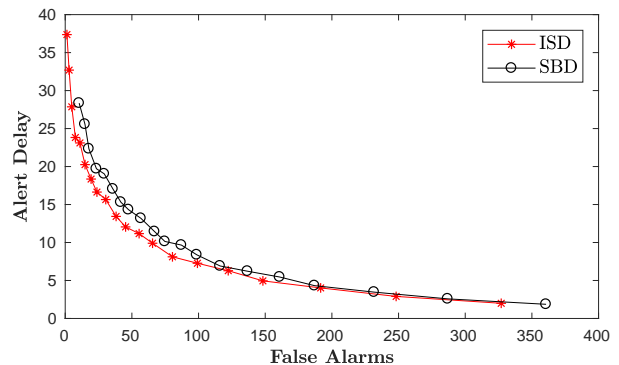}
 \caption{The mean alert delay versus the mean number of false alarms of the ISD and the SBD optimal stopping rules for range of different thresholds $h_s$.  The maximum standard error of the delays is $0.024$. Our ISD rule appears to outperform the SBD rule over a range of different thresholds.}
\label{fig:diff}
\end{center}
\end{figure}

We simulated an intermittent signal $X_k$, the measurements $y_k$, and considered the ISD and SBD rules as described in the previous simulation study.
We compared the performance of the ISD and SBD rules over a range of different thresholds $h_s$ to examine the trade off between the false alarms and the alert delay (AD). For a set threshold we applied both rules for $1000$ Monte Carlo cases and determined the mean AD and mean number of false alarms.  Figure \ref{fig:diff} shows a comparison of the two rules. Perhaps unsurprisingly, the ISD rule appears to outperform the SBD rule over a range of different ADs and false alarms. The maximum standard error of the delays shown in the figure is $0.024$. Additionally, our ISD rule has a theoretical optimality guarantee for this class of intermittent signals while the SBD rule does not.

\subsection{Performance of CME filter for state occupation time}
In our final study we simulated a intermittent signal  $X_k$ 
with transition probabilities $\rho = 0.001$, $a = 0.999$. The measurements $y_k$ were generated as above, except we tested a range of different variances $\sigma^2$. We bounded our PFA with a threshold of $h_s = 0.7$  and applied our ISD rule \eqref{eqn:stop} as above and our occupation time CME filter $\hat{\mathcal{O}}_k^2$ \eqref{eqn:stateocc}	for $1000$ Monte Carlo cases to determine the mean AD.  Figure \ref{fig:occ} shows the mean AD estimated by our proposed occupation time CME filter $\hat{\mathcal{O}}_k^2$ \eqref{eqn:stateocc}	 compared to the mean AD achieved by our proposed ISD rule for a range of different variances $\sigma^2$. The maximum standard error of the delays shown in the figure is $0.28$. Figure \ref{fig:occ} illustrates that the occupation time CME filter $\hat{\mathcal{O}}_k^2$ \eqref{eqn:stateocc}	provides an under estimate for the mean AD which improves as the variance $\sigma^2$ decreases. 

\begin{figure}
\begin{center}
\includegraphics[scale=0.7]{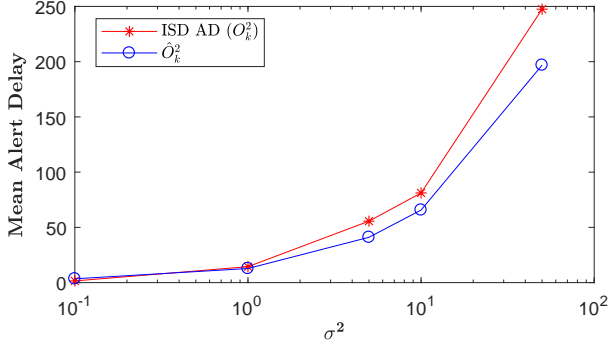}
  \caption{The mean estimated delay time from our proposed occupation time CME filter $\hat{\mathcal{O}}^2_k$ compared to the mean AD achieved by our proposed ISD approach for a range of different variances $\sigma^2$.  The maximum standard error of the delays is $0.28$. Our occupation time CME filter provides an under estimate for the mean AD which improves as the variance $\sigma^2$ decreases.}
\label{fig:occ}
\end{center}
\end{figure}

\section{Application: Vision Based Aircraft Detection}\label{sec:app}
In this section we examine the performance of our ISD rule \eqref{eqn:stop} in the important vision based aircraft detection application. We aim to quickly detect, with low false alarms, an aircraft on a near collision course after it visually emerges. 

Due to the low signal to noise ratio measurements in this application, achieving an effective representation of the dynamics of aircraft emergence is  important. 
Previous work utilising \textit{ad hoc} maximum likelihood detection approaches observed the need for ergodic representations of the aircraft emergence dynamics, which motivated the (physically unrealistic) image boundary transition wrapping used in current approaches \cite{Lai2011,Lai2013}.
It is not clear how classic Bayesian quickest detection might be used in this application due to its absorbing state (i.e. non ergodic representations).


We cast the vision based aircraft detection problem as a quickest ISD problem and then compare the performance of the resulting ISD rule to a baseline detection system on the basis of experimentally captured in-flight image sequences. The aircraft sequences are between two fixed wing aircraft; the data collection aircraft was a ScanEagle UAV  and the other aircraft was a Cessna 172 (see \cite{Bratanov2017} for details of flight experiments). 

\subsection{HMM Aircraft Dynamics}
Consider a single aircraft which we aim to detect at distances where it is (potentially) visually apparent at a single pixel in an image frame. 
For $k \geq 0$,  we introduce a new Markov chain with a state for each of the aircraft's possible $N$ pixel locations.  
We introduce an extra state to denote when the aircraft is not visually apparent (NVA)  anywhere in the image frame.  Let us denote this Markov chain as $Z_k \in \{\mathcal{E}_1, \mathcal{E}_2, \dots, \mathcal{E}_N, \mathcal{E}_{N+1}\}$
where for $1 \leq i \leq N$, $\mathcal{E}_i$ corresponds to the aircraft being visually apparent at the $i$th pixel and $\mathcal{E}_{N+1}$ corresponds to the aircraft not being visually apparent (i.e, in the NVA state). 

Between consecutive frames the aircraft can transition between different Markov states. The likelihood of state transitions depends on expected aircraft motion and are modelled by the HMM transition probabilities $\mathcal{A}^{i,j}$ for $1 \leq i,j \leq N+1$. Within the image, the possible aircraft inter-frame motion can be represented by a transition patch (see \cite{Lai2013} for detailed explanation of patches). State transitions that would cross the image boundary will transition to the NVA state. An aircraft located in the NVA state is able to transition to any pixel in the image (that is, the aircraft can visually emerge anywhere as it approaches from a distance). 

\subsection{Aircraft Observations}
At each time $k>0$ we obtain a noise corrupted morphologically processed greyscale images of an aircraft $y_k$, as in \cite{Lai2011,Lai2013}. 
We denote the measurement of the $i$th pixel at time $k$ as $y^i_k$.
Following \cite{molloy2014} we let $p(y^i_k)$ denote the probability density of pixels occupied by an aircraft and $q(y^i_k)$ denote the probability density of pixels not occupied by an aircraft. 
That is, for $1 \leq i \leq N$ 
\begin{equation*}
y^i_k \sim
\begin{cases} 
      p(y^i_k) & \text{for }  Z_k=\mathcal{E}_i \\
      q(y^i_k) &  \text{for }  Z_k \neq \mathcal{E}_i. \\
\end{cases}
\end{equation*}
Recalling that we consider single pixel sized aircraft, we assume that these densities are statistically independent in the sense that $p(y_k^m,y_k^n) = p(y_k^m)p(y_k^n)$ and $q(y_k^m,y_k^n) = q(y_k^m)q(y_k^n)$ for $1 \leq m,n,\leq N+1, m\neq n$. Hence we have  that the probability of receiving an image $y_k$ when the aircraft is in the $i$th pixel is
\begin{eqnarray*}
b^i(y_k) &\triangleq& p(y_k^i) \prod^{N}_{\substack{j=1 \\ j\neq i}}q(y^j_k) \nonumber \\
 &=& \frac{p(y_k^i)}{q(y_k^i)} \prod^{N}_{j=1} q(y^j_k).
\end{eqnarray*}
Noting that $\prod^{N}_{j=1} q(y^j_k)$ is a common factor for all $1 \leq i \leq N$, we can instead consider the unnormalised  $\bar{b}^i(y_k) = \frac{p(y_k^i)}{q(y_k^i)}$. However, as $p(y_k^i)$ and $q(y_k^i)$ are not known \textit{a priori}, we follow \cite{Lai2013} and use the approximation $\bar{b}^i(y_k) = y^i_k+1$ for $1\leq i \leq N$. When the aircraft is in the NVA state, i.e. for $i=N+1$, then the whole image would consist of noise with no aircraft $b^{N+1}(y_k) = \prod^{N}_{j=1} q(y^j_k)$, giving the unnormalised $\bar{b}^{N+1}(y_k) = 1$. 
Our diagonal matrix of (unnormalised) output densities is then given by
 \begin{equation*}
\mathcal{B}^{ij}(y_k) = 
 \begin{cases}
 \bar{b}^i(y_k) & \text{for } i=j\\
 0 & \text{for } i \neq j\\
 \end{cases}
 \end{equation*}
for $1 \leq i,j \leq N+1$. 

\subsection{Applying our ISD Optimal Stopping Rule}
Recall that our goal is to  quickly detect  when an aircraft emerges in an image sequence, specifically,  when the aircraft leaves the NVA state and appears in any of the pixels in the image frame.
Hence, we seek to design a rule $\tau \geq 0$ for stopping that minimises the following cost criterion
\begin{equation}\label{eqn:newcost}
\mathcal{J}(\tau) = \mathbb{E} \Bigg[c \sum_{i=1}^N \sum_{\ell=0}^{\tau -1} \langle Z_\ell,\mathcal{E}_i \rangle + \langle Z_\tau,\mathcal{E}_{N+1} \rangle  \Bigg],
\end{equation}
which represents our desire to detect when the aircraft appears at any pixel as quickly as possible whilst avoiding false alarms. 

Consider two possible detection states: a no aircraft (normal) state $e_1$ and an aircraft (anomalous) state $e_2$. We can construct these states by  equating $ \langle X_k, e_{1} \rangle  = \langle Z_k, \mathcal{E}_{N+1} \rangle$ and $  \langle X_k , e_2 \rangle =\sum_{i=1}^N \langle Z_k, \mathcal{E}_i \rangle$ through aggregating our first $N$ image states (see \cite{Sonin1999} for information on state aggregation). The cost criterion \eqref{eqn:newcost} now reduces to our ISD cost criterion \eqref{eqn:cost} allowing the use of our ISD rule  \eqref{eqn:stop}
\begin{equation*}
\tau^* = \inf\{k \geq 0 : \zeta_k \geq h_c\}, 
\end{equation*}
where $\zeta_k = 1- \hat{Z}^{N+1}_k$  and $h_c$ is a threshold chosen to trade off the alert delay and probability of false alarm. The CME $\hat{Z}^{N+1}_k$ of the NVA state  can be efficiently calculated via the HMM filter \cite{elliott1995}
\begin{equation*}
\hat{Z}_k = \mathcal{N}_k\mathcal{B}(y_k) \mathcal{A}\hat{Z}_{k-1},
\end{equation*}
where 
\begin{equation*}
\mathcal{N}^{-1}_k  = \langle \underline{1}, \mathcal{B}(y_k) \mathcal{A}\hat{Z}_{k-1}  \rangle.	
\end{equation*}
 


 
\subsection{Performance Study}
In this section we evaluate our proposed ISD rule in an application study on an experimentally captured flight dataset.  We will compare the performance of our proposed rule to a baseline system developed in \cite{Lai2013}. We will denote this baseline rule smoothed normalisation thresholding (SNT-4).  We highlight that the baseline SNT-4 rule employs a filter bank (4 HMM filters) while our proposed ISD rule just uses a single filter 
(with $\mathcal{A}$ having the patch from \cite{Lai2013} that allows motion in the up direction for transitions within the image and an equal probability of  $0.1/N$ for transitions from the NVA state to each pixel in the image).

Detection performance will be evaluated on the 15 head-on near collision course encounters reported in \cite{Bratanov2017}  where we have maintained their numbering convention for comparison purposes.



\subsubsection{Detection Range Study}
Note that detection range and false alarm performance varies with the choice of the threshold parameters. 
Here, we will identify the lowest thresholds $h_c$ for each algorithm that achieve zero false alarms (ZFAs) for this dataset. We will compare the two rule on the basis of their resulting ZFA detection ranges (the ability to achieve low false alarm rates is consistent with findings in \cite{Lai2011, Lai2013}). In practice, detection thresholds could be adaptively selected on the basis of scene difficulty such as proposed in  \cite{molly2017}.

The resulting ZFA detection ranges are presented in Figure \ref{fig:applicationComp}.  The mean detection distance and standard error was $2227$m and $52$m for the ISD rule and  $2076$m and $42$m for the baseline SNT-4 rule. Our ISD rule improved detection ranges relative to the baseline SNT-4 rule by a mean distance of $151$m. A paired-sample t-test shows at a significance level of 0.05  that our proposed ISD rule performs at least $3.6\%$ ($75$m) better than the baseline SNT-4 rule. 

\begin{figure}[t!]
\begin{center}
\includegraphics[scale=0.55]{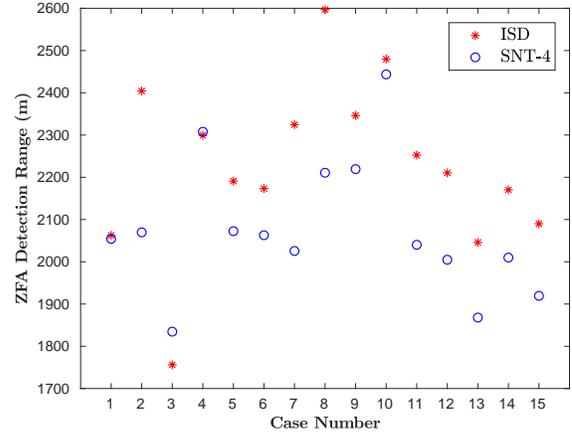}
\caption{A comparison of the proposed ISD rule compared to the baseline  SNT-4 from \cite{Lai2013} for all 15 cases presented in \cite{Bratanov2017}. The mean detection distance and standard error was $2227$m and $52$m for the ISD rule and  $2076$m and $42$m for the SNT-4 rule.  } 
\label{fig:applicationComp}
\end{center}
\end{figure}

\subsubsection{System Operating Characteristics Analysis}
\begin{figure}[t!]
\begin{center}
\includegraphics[scale=0.53]{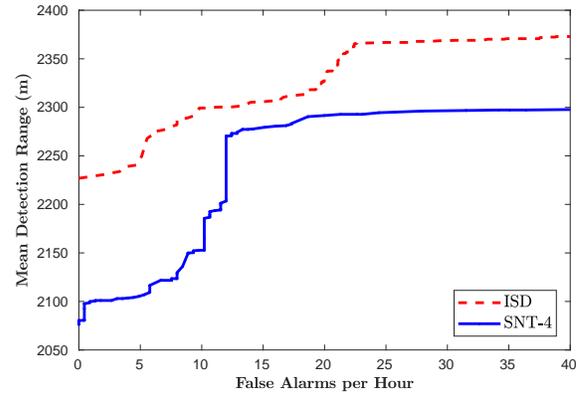}
\caption{The mean detection ranges and mean false alarm rates for our proposed ISD rule compared to the baseline SNT-4 rule. The maximum standard error of the mean detection ranges is $72$m for our proposed ISD rule and  $108$m for the baseline SNT-4 rule.} 
\label{fig:socCuve}
\end{center}
\end{figure}

We next composed system operating characteristic (SOC) curves for our proposed ISD rule and baseline SNT-4 rule (SOC curves examine detection range and false alarm performance for different thresholds). Figure \ref{fig:socCuve} presents the mean detection range  for all 15 cases versus the mean false alarms per hour. The maximum standard error of the mean detection ranges is $72$m for our proposed ISD rule and  $108$m for the baseline SNT-4 rule. Figure \ref{fig:socCuve} illustrates longer detection ranges for our proposed system whilst maintaining lower false alarm rates across all tested thresholds. 

\subsection{Advanced Detection Rule Study}
We compare the performance of our proposed ISD rule with 4 other detection rules.
We modified the baseline SNT-4 rule to have individual thresholds for each of the filters in the bank (we denote this SNT-4I).
We also considered a 4 filter bank version of the  normalisation change detection (NCD) approach \cite{James2017a}  with individual thresholds (we denote this NCD-4I). 
Finally, we implemented a 4 filter bank version of our proposed ISD rule with individual thresholds (we denote this ISD-4I). 

Table \ref{tbl:RQ} presents the mean detection ranges and standard errors for the five compared detection rules. We highlight that our proposed ISD and ISD-4I  rules illustrate longer mean detection ranges than all other rules. Additionally our proposed ISD rules have the benefit of not requiring an estimate of aircraft and non-aircraft densities (this is required in the NCD-4I rule).
\begin{table}
\caption{The mean ZFA detection ranges and standard errors for all 15 cases.  } 
\begin{center}
\begin{tabular}{ |c|c| } 
\hline
Detection rule & Mean ZFA detection range (m) \\
 \hline
SNT-4 & $2076 \pm 42$  \\ 
SNT-4I & $2215 \pm 62$ \\
NCD-4I & $2216 \pm 62$ \\
ISD & $2227 \pm 52$  \\ 
ISD- 4I & $2376 \pm 82$ \\
 \hline
\end{tabular}
\label{tbl:RQ}
\end{center}
\end{table}

\section{Conclusion}
In this paper we examined the problem of quickly detecting changes  in an intermittent signal. We first posed the quickest ISD problem and established an optimal ISD rule.  We then developed techniques and bounds for characterising the performance of our ISD rule. 
Finally, we investigated the performance of our ISD rule in both a simulation study and in the important vision based aircraft detection application. 
We were able to show that our ISD rule improves detection performance by at least $3.6\%$, at a significance level of $0.05$,  relative to the current state of the art vision based aircraft detection technique.  



\appendix
\subsection*{Proof of Lemma 2}
We first introduce some measure change concepts, see \cite{elliott1995} for more details. 
Let us define a new probability measure $\bar{\mathbb{P}}$ on $(\Omega,\mathcal{F})$ under which $y_k$ becomes a sequence of i.i.d. random variables with probability density function $\psi(\cdot)$. Let $\bar{\mathbb{E}}[\cdot]$ denote the expectation operation defined by $\bar{\mathbb{P}}$. Let $\mathcal{F}_k$ denote the complete filtration generated by $(X_{[0,k]},y_{[1,k]})$. We can define a measure change between $\mathbb{P}$ and $\bar{\mathbb{P}}$ via the Radon-Nikodym derivative 
$\left. \frac{d\mathbb{P}}{d\bar{\mathbb{P}}} \right|_{\mathcal{F}_k}= \bar{\Lambda}_k$ 
as follows (see  \cite{elliott1995, Xie}) 
\begin{equation*}
 \bar{\Lambda}_k = \prod^k_{\ell=1} \bar{\lambda}_\ell, \quad \bar{\lambda}_\ell = \frac{\begin{bmatrix}
 f^1(y_\ell)  & f^2(y_\ell)
 \end{bmatrix}X_\ell }{\psi(y_\ell)}.
\end{equation*}
We can now introduce an unnormalised CME of the state, $\bar{X}_k \triangleq \bar{\mathbb{E}} \Big[\bar{\Lambda}_k X_k \Big| y_{[1,k]} \Big]$
which is related to our normalised estimate via the conditional Bayes theorem \cite[Thm 2.3.2]{elliott1995} as follows
\begin{equation}\label{eqn:xbayes}
\hat{X}_k = \frac{\bar{X}_k}{\bar{\mathbb{E}}\big[ \bar{\Lambda}_k \big| y_{[1,k]}\big]}.
\end{equation}

Similarly, we can introduce an unnormalised CME of the occupation time ending in state $X_k$, $X_k\mathcal{O}^{i}_k$ as
$\bar{\mathcal{O}}^{i,X}_k = \bar{\mathbb{E}} \Big[ \bar{\Lambda}_k X_k {\mathcal{O}}^{i}_{k}  \Big| y_{[1,k]} \Big]$
which is related to our normalised estimate via the conditional Bayes theorem as
\begin{equation}\label{eqn:obayes}
\hat{\mathcal{O}}^{i,X}_k= \frac{\bar{\mathcal{O}}^{i,X}_k }{\bar{\mathbb{E}}\big[ \bar{\Lambda}_k\big| y_{[1,k]}\big]}.
\end{equation}

Before we present our main argument we note that martingale increment properties of $V_k$ gives that $\bar{\mathbb{E}} \left[ \bar{\Lambda}_{k-1} V_k \langle X_{\ell-1},e_i \rangle \left| X_{[0,k-1]},y_{[1,k]} \right] \right. =0$ $\bar{\mathbb{P}} \; a.s.$ for any $\ell \le k$ and any $i$. Then using this result within the tower rule for conditional expectations \cite[pg. 331]{elliott1995} shows that, for any $\ell \le k$ and any $i$,
\begin{equation} \label{eqn:vzero}
\bar{\mathbb{E}} \Big[\bar{\Lambda}_{k-1}  V_k  \langle X_{\ell-1},e_i \rangle \Big| y_{[1,k]} \Big] = 0 \;  \quad  \bar{\mathbb{P}} \; a.s..
\end{equation}

Simple algebra also gives
\begin{equation} \label{eqn:lstep}
\bar{\lambda}_k X_k =\frac{B(y_k)}{\psi(y_k)}  X_k.
\end{equation}


For $i \in \{1,2\}$, we can rewrite the unnormalised estimate $\bar{\mathbb{P}} \; a.s.$ as
\begin{equation*}
\begin{split}
\bar{\mathcal{O}}^{i,X}_k =& \bar{\mathbb{E}} \Big[\bar{\Lambda}_k X_k {\mathcal{O}}^{i}_{k}  \Big| y_{[1,k]} \Big]\\
 =& \bar{\mathbb{E}} \Big[\bar{\lambda}_k \bar{\Lambda}_{k-1} X_k  {\mathcal{O}}^{i}_{k}  \Big| y_{[1,k]} \Big]\\
 =& \frac{B(y_k)}{\psi(y_k)} \bar{\mathbb{E}} \Big[  \bar{\Lambda}_{k-1}  X_k {\mathcal{O}}^{i}_{k}  \Big| y_{[1,k]} \Big]\\
  =& \frac{B(y_k)}{\psi(y_k)} \bar{\mathbb{E}} \Big[  \bar{\Lambda}_{k-1}  (A X_{k-1} + V_k)  {\mathcal{O}}^{i}_{k} \Big| y_{[1,k]} \Big]\\
    =& \frac{B(y_k)}{\psi(y_k)} \bar{\mathbb{E}} \Big[  \bar{\Lambda}_{k-1}  A X_{k-1} \left( {\mathcal{O}}^{i}_{k-1} + \langle X_{k-1},e_i \rangle \right) \Big| y_{[1,k]} \Big]\\
   =& \frac{B(y_k)}{\psi(y_k)} A \bar{\mathbb{E}} \Big[  \bar{\Lambda}_{k-1} \left( X_{k-1} \mathcal{O}^{i}_{k-1} + {X}^i_{k-1} e_i \right)  \Big| y_{[1,k]} \Big]
  \\
     =& \frac{B(y_k)}{\psi(y_k)} A \bar{\mathbb{E}} \Big[  \bar{\Lambda}_{k-1} \left(X_{k-1} \mathcal{O}^{i}_{k-1} + {X}^i_{k-1} e_i \right) \Big| y_{[1,k-1]} \Big]
  \\
  =& \frac{B(y_k)}{\psi(y_k)} A \left( \bar{\mathcal{O}}^{i,X}_{k-1} + \bar{X}^i_{k-1} e_i \right)
\end{split}
\end{equation*}
where in the third line we have used (\ref{eqn:lstep}), in the fourth line we have used  (\ref{eqn:xprocess}), in the fifth line we have used the definition of $\mathcal{O}^{i}_{k}$ and 
(\ref{eqn:vzero}), in the sixth line we have used that $X_{k-1} \langle X_{k-1},e_i \rangle = {X}_{k-1}^i e_i$. We then use that $y_k$ is {\it i.i.d} under $\bar{\mathbb{E}}$, and finally the definitions of the unnormalised estimates.

From \eqref{eqn:xbayes} and \eqref{eqn:obayes} we note there is a common normalisation factor. We can now write our CME as 
\begin{equation*}
\hat{\mathcal{O}}^{i,X}_k= N_k B(y_k) A (\hat{\mathcal{O}}^{i,X}_{k-1} + \hat{X}^i_{k-1} e_i).
\end{equation*}
Finally,   an inner product gives our occupation time CME
\begin{equation*}
\hat{\mathcal{O}}^{i}_k = \langle \underline{1}, \hat{\mathcal{O}}^{i,X}_k \rangle.
\end{equation*}
This completes the proof. 
\hfill \IEEEQEDopen

\subsection*{Proof of Lemma 3}
Let us define $U_{1,k} \triangleq B(y_k)AB(y_{k-1})A,...,B(y_1)A$ and $N_{1,k}(\hat{X}_0) \triangleq \left( 1' U_{1,k}\hat{X}_0\right)^{-1}.$ Note that we can write the CME as $\hat{X}_k(\hat{X}_0) = N_{1,k}(\hat{X}_0) U_{1,k} \hat{X}_0$.

For $k  > 0$ we can write our occupation time filter $\hat{\mathcal{O}}^{i,X}_k $ in terms of an initial condition $\hat{X}_0$, as follows
\begin{equation*}
\begin{split}
\hat{\mathcal{O}}^{i,X}_k(\hat{X}_0) &= N_{1,k}(\hat{X}_0) U_{1,k}   \hat{\mathcal{O}}^{i,X}_0  \\
+ \sum^k_{j=1}& N_{j,k}(\hat{X}_{j-1}(\hat{X}_0)) U_{j,k} \hat{X}_j^r( \hat{X}_0) e_i.
\end{split}
\end{equation*}

Noting the initialisation  $\hat{\mathcal{O}}^{i,X}_0= \underline{0}$,  for any $i$,  we have
\begin{equation*}
\begin{split}
\big| \hat{\mathcal{O}}^{i,X}_k&(\hat{X}_0) - \hat{\mathcal{O}}^{i,X}_k(\check{X}_0) \big| \\
\leq & \sum^k_{j=1}  \mathcal{C}_{j,k}(\hat{X}_0,\check{X}_0) 
\end{split}
\end{equation*}
where $  \mathcal{C}_{j,k}(\hat{X}_0,\check{X}_0) \triangleq \big| N_{j,k}(\hat{X}_{j-1}(\hat{X}_0)) U_{j,k} \hat{X}^i_j( \hat{X}_0) e_i $ $- N_{j,k}(\hat{X}_{j-1}(\check{X}_0)) U_{j,k} \hat{X}^i_j( \check{X}_0) e_i \big|$.
Under our lemma assumption, for $k>j$, we can now write
\begin{equation*}
 \mathcal{C}_{j,k}(\hat{X}_0,\check{X}_0)  \leq \beta \left( \left| \hat{X}^i_j( \hat{X}_0) e_i - \hat{X}^i_{j}( \check{X}_0)) e_i \right|, k-j \right).
\end{equation*}

Given that $\beta \left( \left| \hat{X}^i_j( \hat{X}_0) e_i - \hat{X}^i_{j}( \check{X}_0)) e_i \right|, k-j \right) < \beta \left( 1, k-j \right)$,  we note that for any given $\delta$ there  is a $H$ such that for sufficiently large $k$ we can write
 \begin{equation*}
\mathcal{C}_{j,k}(\hat{X}_0,\check{X}_0)  < 
 \begin{cases}
 \delta & \text{for } k-j \geq H\\
 1 & \text{for } k-j < H.\\
 \end{cases}
 \end{equation*}

Finally, because there are $H$ terms that are bounded by $1$ and $k-H$ terms bounded by $\delta$, we can bound our error as
\begin{equation*}
\left| \hat{\mathcal{O}}^{r,X}_k(\hat{X}_0) - \hat{\mathcal{O}}^{r,X}_k(\check{X}_0) \right| \leq \delta(k-H) +H.
\end{equation*}
Dividing by $k$ and using that $\delta <1$ gives
\begin{equation*}
\mathcal{R}^E_k(\hat{X}_0,\check{X}_0) \leq \delta + \frac{H}{k},
\end{equation*}
and this completes the proof. \hfill \IEEEQEDopen


\bibliographystyle{IEEEtran}
\bibliography{IEEEabrv,ref}

\begin{thebibliography}{10}
\providecommand{\url}[1]{#1}
\csname url@samestyle\endcsname
\providecommand{\newblock}{\relax}
\providecommand{\bibinfo}[2]{#2}
\providecommand{\BIBentrySTDinterwordspacing}{\spaceskip=0pt\relax}
\providecommand{\BIBentryALTinterwordstretchfactor}{4}
\providecommand{\BIBentryALTinterwordspacing}{\spaceskip=\fontdimen2\font plus
\BIBentryALTinterwordstretchfactor\fontdimen3\font minus
  \fontdimen4\font\relax}
\providecommand{\BIBforeignlanguage}[2]{{%
\expandafter\ifx\csname l@#1\endcsname\relax
\typeout{** WARNING: IEEEtran.bst: No hyphenation pattern has been}%
\typeout{** loaded for the language `#1'. Using the pattern for}%
\typeout{** the default language instead.}%
\else
\language=\csname l@#1\endcsname
\fi
#2}}
\providecommand{\BIBdecl}{\relax}
\BIBdecl

\bibitem{James2017b}
J.~James, T.~L. Molloy, and J.~J. Ford, ``{Quickest Detection of Intermittent
  Signals with Estimated Anomaly Times},'' in \emph{The Asian Control
  Conference, ASCC 2017}.\hskip 1em plus 0.5em minus 0.4em\relax IEEE, Dec (In
  press).

\bibitem{Hwang}
I.~Hwang, S.~Kim, Y.~Kim, and C.~E. Seah, ``{A Survey of Fault Detection,
  Isolation, and Reconfiguration Methods},'' \emph{IEEE Transactions on Control
  Systems Technology}, vol.~18, no.~3, pp. 636--653, May 2010.

\bibitem{Tartakovsky2012}
A.~G. Tartakovsky, A.~S. Polunchenko, and G.~Sokolov, ``{Efficient Computer
  Network Anomaly Detection by Changepoint Detection Methods},'' Dec 2012.

\bibitem{Lai2013}
J.~Lai, J.~J. Ford, L.~Mejias, and P.~O'Shea, ``{Characterization of sky-region
  morphological-temporal airborne collision detection},'' \emph{Journal of
  Field Robotics}, vol.~30, no.~2, pp. 171--193, Mar 2013.

\bibitem{Tartakovsky2014}
A.~Tartakovsky, I.~V. I.~V. Nikiforov, and M.~M. Basseville, \emph{{Sequential
  analysis : hypothesis testing and changepoint detection}}, ser. Chapman {\&}
  Hall/CRC Monographs on Statistics {\&} Applied Probability.\hskip 1em plus
  0.5em minus 0.4em\relax Chapman {\&} Hall/CRC, Taylor and Francis Group, Aug
  2014.

\bibitem{Roychoudhury}
I.~Roychoudhury, G.~Biswas, and X.~Koutsoukos, ``A bayesian approach to
  efficient diagnosis of incipient faults,'' in \emph{in Proc. 17th Int.
  Workshop Principles of Diagnosis}, Jun 2006, pp. 243--250.

\bibitem{Polunchenko2011}
A.~S. Polunchenko and A.~G. Tartakovsky, ``State-of-the-art in sequential
  change-point detection,'' \emph{Methodology and computing in applied
  probability}, vol.~14, no.~3, pp. 649--684, Sep 2012.

\bibitem{Broder}
B.~Broder and S.~Schwartz, ``Quickest detection of transients,'' in
  \emph{Proceedings. 1991 IEEE International Symposium on Information Theory},
  Jun 1991, pp. 356--356.

\bibitem{Guepie2017}
B.~K. Guepie, L.~Fillatre, and I.~Nikiforov, ``{Detecting a Suddenly Arriving
  Dynamic Profile of Finite Duration},'' \emph{IEEE Transactions on Information
  Theory}, pp. 1--1, May 2017.

\bibitem{2017Zou}
S.~{Zou}, G.~{Fellouris}, and V.~V. {Veeravalli}, ``{Quickest Change Detection
  under Transient Dynamics: Theory and Asymptotic Analysis},'' \emph{ArXiv
  e-prints}, Nov 2017.

\bibitem{Lai2011}
J.~Lai, L.~Mejias, and J.~J. Ford, ``{Airborne vision-based collision-detection
  system},'' \emph{Journal of Field Robotics}, vol.~28, no.~2, pp. 137--157,
  Mar 2011.

\bibitem{molloy2017below}
T.~L. Molloy, J.~J. Ford, and L.~Mejias, ``Detection of aircraft below the
  horizon for vision-based detect and avoid in unmanned aircraft systems,''
  \emph{Journal of Field Robotics}, 2017.

\bibitem{James2017a}
J.~James, J.~J. Ford, and T.~L. Molloy, ``Change detection for undermodelled
  processes using mismatched hidden markov model test filters,'' \emph{IEEE
  Control Systems Letters}, vol.~1, no.~2, pp. 238--243, Oct 2017.

\bibitem{ford2011}
J.~J. Ford, V.~A. Ugrinovskii, and J.~S. Lai, ``An infinite-horizon robust
  filter for uncertain hidden markov models with conditional relative entropy
  constraints,'' in \emph{Australian Control Conference}, University of
  Melbourne, Melbourne, VIC, Nov 2011.

\bibitem{Shiryaev}
A.~N. Shiryaev, ``On optimum methods in quickest detection problems,''
  \emph{Theory of Probability \& Its Applications}, vol.~8, no.~1, pp. 22--46,
  1963.

\bibitem{krishnamurthy2016}
V.~Krishnamurthy, \emph{Partially Observed Markov Decision Processes}.\hskip
  1em plus 0.5em minus 0.4em\relax Cambridge University Press, 2016.

\bibitem{elliott1995}
R.~Elliott, L.~Aggoun, and J.~Moore, \emph{Hidden Markov Models: Estimation and
  Control}, ser. Applications of mathematics.\hskip 1em plus 0.5em minus
  0.4em\relax Springer-Verlag, 1995.

\bibitem{molloy2016}
T.~L. Molloy and J.~J. Ford, ``Asymptotic minimax robust quickest change
  detection for dependent stochastic processes with parametric uncertainty,''
  \emph{IEEE Transactions on Information Theory}, vol.~62, no.~11, pp.
  6594--6608, 2016.

\bibitem{Shue1998}
L.~Shue, D.~Anderson, and S.~Dey, ``{Exponential stability of filters and
  smoothers for hidden Markov models},'' \emph{IEEE Transactions on Signal
  Processing}, vol.~46, no.~8, pp. 2180--2194, 1998.

\bibitem{Bratanov2017}
D.~Bratanov, L.~Mejias, and J.~J. Ford, ``{A vision-based sense-and-avoid
  system tested on a ScanEagle UAV},'' in \emph{2017 International Conference
  on Unmanned Aircraft Systems (ICUAS)}.\hskip 1em plus 0.5em minus 0.4em\relax
  IEEE, Jun 2017, pp. 1134--1142.

\bibitem{molloy2014}
T.~L. Molloy, J.~J. Ford, and L.~Mejias, ``Looming aircraft threats :
  shape-based passive ranging of aircraft from monocular vision,'' in
  \emph{Australian Conference on Robotics and Automation 2014}, The University
  of Melbourne, Melbourne, VIC, Dec 2014.

\bibitem{Sonin1999}
I.~Sonin, ``{The Elimination algorithm for the problem of optimal stopping},''
  \emph{Mathematical Methods of Operations Research}, vol.~49, no.~1, pp.
  111--123, Mar 1999.

\bibitem{molly2017}
T.~L. Molloy, J.~J. Ford, and L.~Mejias, ``Adaptive detection threshold
  selection for vision-based sense and avoid,'' in \emph{2017 International
  Conference on Unmanned Aircraft Systems (ICUAS 2017)}.\hskip 1em plus 0.5em
  minus 0.4em\relax Miami, FL: IEEE, Jun 2017, pp. 893--901.

\bibitem{Xie}
L.~Xie, V.~A. Ugrinovskii, and I.~R. Petersen, ``{Probabilistic distances
  between finite-state finite-alphabet hidden Markov models},'' \emph{IEEE
  Transactions on Automatic Control}, vol.~50, no.~4, pp. 505--511, Apr 2005.

\end{thebibliography}


\end{document}